\documentclass[a4paper,11pt]{article}
\usepackage[T1]{fontenc}
\usepackage[utf8]{inputenc}
\usepackage{lmodern}
\usepackage[margin=2.5cm]{geometry}
\usepackage{amsmath,amssymb,amsthm,mathtools}
\usepackage{enumitem}
\usepackage{xurl}
\usepackage{algorithm}
\usepackage{algpseudocode}
\usepackage[hidelinks]{hyperref}

\newcommand{\WPC}{\operatorname{WPC}}
\newcommand{\FWPC}{\overline{\operatorname{WPC}}}

\newcommand{\poly}{\operatorname{poly}}
\newcommand{\cP}{\mathcal{P}}
\newcommand{\cE}{\mathcal{E}}
\theoremstyle{plain}

\newtheorem{theorem}{Theorem}[section]
\newtheorem{lemma}[theorem]{Lemma}
\newtheorem{proposition}[theorem]{Proposition}
\newtheorem{corollary}[theorem]{Corollary}
\theoremstyle{remark}

\numberwithin{equation}{section}
\title{The Complexity of Weak Partition Connectivity in Hedgegraphs}
\author{\small Yuanhao Wang$\,^{\rm a}$\quad\quad Wei Wang$\,^{\rm a}$\footnote{Corresponding author. Email address: wang\_weiw@163.com}\\
\small $^{\rm a}\,$School of Mathematics and Statistics, Xi'an Jiaotong University, Xi'an, 710049, P. R. China}
\date{}

\begin{document}
\maketitle

\begin{abstract}
We prove that deciding whether the weak partition connectivity (WPC) of an unweighted hedgegraph is at most a given integer is NP-complete, resolving an open question about its computational complexity. Hardness holds even for connected full-support split systems, where each hedge consists of two nonempty, vertex-disjoint hyperedges covering all vertices. We also obtain hardness on connected edge-colored simple graphs in which each color class is a spanning forest with exactly two nontrivial trees. On both classes, exact computation of fractional WPC is NP-hard, whereas hedge connectivity is computable exactly in polynomial time. A binary matrix formulation yields the reduction and a deterministic polynomial-time exact algorithm under a linear intersection condition on row supports after normalization and column compression. In particular, both integer and fractional WPC are computable in polynomial time when these supports have size at most two, whereas NP-hardness holds already when their size is bounded by three. For every fixed $b$, we give deterministic and randomized efficient polynomial-time approximation schemes on hedgegraphs in which each hedge has at most $b$ connected components, counting isolated vertices. Both schemes return one partition that simultaneously approximates integer and fractional WPC. Unless $\mathrm P=\mathrm{NP}$, neither objective admits a deterministic FPTAS, even on the two hard instance classes.
\end{abstract}


\section{Introduction}\label{sec:introduction}
\subsection{Background and related work}\label{subsec:background}

A hedgegraph groups edges or hyperedges into \emph{hedges}; deleting a hedge removes all its constituent hyperedges. This model captures the simultaneous failure of edges of the same type. For example, each color class in an edge-colored graph can be viewed as a hedge. Such models arise in network survivability with shared-risk resources and in labeled connectivity problems~\cite{r9,r24}. Hedge connectivity asks for the minimum number of hedges whose deletion disconnects the underlying hypergraph. Ghaffari, Karger, and Panigrahi~\cite{r15} gave randomized approximation algorithms and a quasi-polynomial-time exact algorithm for this problem. Randomized polynomial-time constructions of hedge-cut sparsifiers have also been obtained through CSP sparsification~\cite{r28}. Subsequent work investigated structural restrictions and parameterized algorithms. Randomized polynomial-time algorithms are known when the span, that is, the number of component hyperedges in each hedge, is bounded by a constant~\cite{r8,r11}. General Hedge Cut admits a randomized fixed-parameter algorithm~\cite{r10}, and quasi-polynomial lower bounds are known under the Exponential Time Hypothesis (ETH)~\cite{r18,r23}. 

Chandrasekaran et al.~\cite{r5} developed a polymatroidal approach to hedgegraphs and studied two partition-based connectivity measures. They gave a polynomial-time algorithm for partition connectivity and a deterministic $O(\log |V|)$-approximation for hedge connectivity. They also gave a randomized polynomial-time sparsification algorithm that approximately preserves the total weight of hedges crossing every vertex partition. They asked whether weak partition connectivity (WPC) can be computed in polynomial time. We show that it cannot unless $\mathrm{P}=\mathrm{NP}$. More recently, Chandrasekaran, Chekuri, and Zhu~\cite{r6} studied the minimum nonempty quotient of a polymatroid given by a value oracle, placing random contraction and hedge-cut results in a common framework. For a fixed number of objectives, their work~\cite{r7} gives a randomized PTAS and randomized quasi-polynomial-time algorithms for hypergraph cuts, together with extensions to hedge cuts.

In binary feature selection, rows represent features and columns represent data points; retaining rows determines the distinct projected data patterns. Bandyapadhyay et al.~\cite{r1} studied feature selection for categorical data clustering. Projection classes also appear in Partial VC Dimension~\cite{r2}, which maximizes the number of distinct traces for a prescribed selection size. Our objective instead balances the number of selected rows against the number of projection types.

Characterizing connectivity through vertex partitions is closely related to the edge-disjoint spanning-tree packing theorems of Nash-Williams and Tutte~\cite{r20,r22} and their hypergraph extensions~\cite{r12,r13}. For ordinary hypergraphs, WPC can be computed in polynomial time via its characterization in terms of rooted orientations~\cite{r12}. In this setting, WPC also appears in list decoding and tree assignments~\cite{r17,r26}, and in the splitting-off approach to Steiner rooted-connected orientations~\cite{r29}. Zhang and Yang~\cite{r25} further study excess obstructions and certificates for tree assignments. Our exact algorithm reduces induced-density optimization to maximum-weight closure~\cite{r21}.

\subsection{Main results}\label{subsec:main-results}

We write $\lambda(G)$ for hedge connectivity, $\WPC(G)$ for integer weak partition connectivity, and $\FWPC(G)$ for its unrounded version. Precise definitions are given in Section~\ref{sec:definitions}. Our main hardness result is the following.

\begin{theorem}\label{thm:main}
Given an explicitly represented unweighted hedgegraph $G$ and an integer $K$, deciding whether $\WPC(G)\le K$ is NP-complete. This remains true when $G$ is connected and every hedge consists of exactly two nonempty, vertex-disjoint hyperedges whose union is $V(G)$.
\end{theorem}

We call this restricted class \emph{full-support split systems}. If there are $m$ hedges and the most frequent split occurs $\mu$ times, then $\lambda(G)=m-\mu$. Thus, on full-support split systems, computing WPC is NP-hard, whereas hedge connectivity is computable in polynomial time. The same separation holds for connected edge-colored simple graphs in which each color class is a spanning forest with exactly two nontrivial trees. Section~\ref{sec:hardness} proves these claims and gives a short enumeration algorithm for minimum hedge cut whenever every hedge has at most a fixed number of components.

The reduction and the exact algorithm use a binary matrix $A\in\{0,1\}^{m\times N}$ whose nonconstant rows represent splits and whose columns represent vertices. If $c_A(R)$ counts the distinct projected columns on a row set $R$, define
\[
\rho(A)=\max_{\varnothing\ne R\subseteq[m]}\frac{|R|}{c_A(R)-1}.
\]
Both WPC values can be expressed in terms of $\rho(A)$: 
\begin{equation}\label{eq:intro-density}
\FWPC(G_A)=m-\rho(A),\qquad \WPC(G_A)=m-\lceil\rho(A)\rceil.
\end{equation}
It turns partition optimization into row selection. Our reduction encodes the independence number of a source graph in $\rho(A)$. Under the linear intersection condition of Section~\ref{sec:exact}, both WPC values and a common optimal partition can be computed in deterministic polynomial time. In particular, $s(A)\le2$ is tractable, whereas the integer-threshold problem is NP-complete for $s(A)=3$. Here $s(A)$ is the minimum possible maximum row-support size after normalization and column compression.

For every fixed $b\ge2$, we give deterministic and randomized efficient polynomial-time approximation schemes (EPTASs) whenever every hedge has at most $b$ connected components on the entire vertex set. Each returns one partition simultaneously approximating integer and fractional WPC. Their respective running times are
\[
2^{O_b(\varepsilon^{-1}\log\log(1/\varepsilon+3))}\poly(\ell)\quad\text{and}\quad 2^{O_b(\varepsilon^{-1})}\poly(\ell),
\]
where $\ell$ is the explicit input length and the polynomial degrees are independent of $b$ and $\varepsilon$. The randomized guarantee holds with probability at least $2/3$. The case $b=2$ includes all full-support split systems. Unless $\mathrm P=\mathrm{NP}$, neither objective admits a deterministic FPTAS, even on the hard split systems.

\subsection{Organization}\label{subsec:organization}

Section~\ref{sec:definitions} gives the input conventions and matrix density formula. Section~\ref{sec:hardness} proves NP-completeness, establishes the separation from hedge connectivity, and transfers hardness to edge-colored simple graphs. Section~\ref{sec:exact} gives the exact algorithm under the linear intersection condition. Section~\ref{sec:approximation} gives the two EPTASs for bounded-component hedgegraphs, recovers split systems as a special case, and excludes a deterministic FPTAS. Section~\ref{sec:conclusion} discusses the remaining questions.

\section{Definitions and the matrix density formula}\label{sec:definitions}

We represent full-support split systems by binary matrices and express their weak partition connectivity in terms of a row-density parameter. We begin with the definitions and input model.

\subsection{Definitions and input conventions}\label{subsec:definitions}

The hedgegraph model, contraction by a vertex partition, hedge connectivity, and integer weak partition connectivity in this subsection follow Chandrasekaran et al.~\cite{r5}. We fix notation and make explicit how repeated hyperedges are counted.

Let $G=(V,\cE)$ be a hedgegraph, where $V$ is the vertex set, $|V|\ge2$, and $\cE$ is the collection of hedges. Each hedge $e\in\cE$ contains finitely many hyperedges, called its \emph{component hyperedges}. Every component hyperedge is nonempty, and distinct component hyperedges within a hedge are vertex-disjoint. This restriction does not apply to hyperedges in different hedges. Each hyperedge belongs to exactly one hedge, and deleting a hedge deletes all its component hyperedges. Forgetting the grouping while retaining all vertices and hyperedges gives the \emph{underlying hypergraph}. Throughout the paper, connectivity refers to connectivity of this underlying hypergraph.

For hedge $e_i$, let $\pi_i$ be the connected-component partition of the spanning hypergraph $(V,\{e_i\})$, with $q_i=|\pi_i|$. This partition includes every isolated vertex as a singleton. We say that there are at most $b$ components per hedge when $q_i\le b$ for every $i$.

Every hedge and component hyperedge has its own identifier. The vertex set of a hyperedge is its \emph{support}. Distinct hyperedges with the same support are called \emph{parallel copies} and are counted separately. For example, if two hedges each contain a hyperedge with vertex set $\{u,v\}$, deleting one hedge does not delete the hyperedge in the other. Repeated hedges are also counted separately.

A cut $(S,V\setminus S)$ divides the vertex set into two nonempty parts, so $\varnothing\ne S\subsetneq V$. A hedge \emph{crosses} the cut if at least one of its component hyperedges contains vertices on both sides. A hedge is counted only once even if several of its component hyperedges cross the cut. Let $d_G(S)$ be the number of crossing hedges. Hedge connectivity is
\[
\lambda(G)=\min_{\varnothing\ne S\subsetneq V}d_G(S).
\]
Deleting all crossing hedges removes every hyperedge joining the two sides. Equivalently, $\lambda(G)$ is the minimum number of hedges whose deletion disconnects the underlying hypergraph.

A vertex partition $\cP$ is a collection of pairwise disjoint nonempty subsets of $V$ whose union is $V$. Its members are called \emph{blocks}, and $|\cP|$ is its number of blocks. Fix a hedge $e$ and retain only its component hyperedges together with all vertices, obtaining $(V,\{e\})$. \emph{Contracting by $\cP$} identifies all vertices in each block into a single new vertex. The image of a hyperedge $h$ consists of the blocks that intersect $h$; the resulting hypergraph is the quotient hypergraph. Distinct component hyperedges may intersect after contraction, so its connected components must be recomputed. Isolated quotient vertices are included in the component count.

Let $\kappa_e(\cP)$ denote the number of connected components in this quotient hypergraph. For $p=|\cP|$, the quantity $p-\kappa_e(\cP)$ measures the connections established by the hedge between blocks: a connected component containing $r$ quotient vertices contributes $r-1$. Summing these contributions gives
\begin{equation}\label{eq:partition-contribution}
D_G(\cP)=\sum_{e\in\cE}\bigl(|\cP|-\kappa_e(\cP)\bigr).
\end{equation}
We define
\begin{equation}\label{eq:wpc-definition}
\FWPC(G)=\min_{\substack{\cP\text{ a partition of }V\\|\cP|\ge2}}\frac{D_G(\cP)}{|\cP|-1}, \qquad \WPC(G)=\lfloor\FWPC(G)\rfloor.
\end{equation}
The quantity $\WPC(G)$ in~\eqref{eq:wpc-definition} is the integer weak partition connectivity defined in the cited work. We call the unrounded optimum $\FWPC(G)$ fractional WPC. Since there are only finitely many partitions, taking the minimum and then rounding down is equivalent to rounding down each score before taking the minimum.

The input explicitly lists all vertices, the component hyperedges of every hedge, and the vertices of every hyperedge. Repeated hedges and hyperedges are listed individually. Apart from the bits needed for identifiers, the input size is polynomially equivalent to
\[
|V|+|\cE|+\sum_{e\in\cE}\sum_{h\in e}|h|.
\]
Here $|h|$ is the number of vertices of $h$, and the final sum counts all listed vertex incidences, counting a vertex separately each time it appears in a different hyperedge.

If $\cE=\varnothing$, then $\lambda(G)=\WPC(G)=\FWPC(G)=0$, and every partition into at least two nonempty blocks is optimal for both WPC objectives. We handle this case separately and assume $m=|\cE|\ge1$ in the matrix formulations and algorithms below.

The following comparison extends the integer bounds in~\cite{r5} to fractional WPC:
\begin{equation}\label{eq:connectivity-bounds}
\left\lfloor\frac{\lambda(G)}2\right\rfloor\le\WPC(G)\le\lambda(G), \qquad \frac{\lambda(G)}2\le\FWPC(G)\le\lambda(G).
\end{equation}
\begin{proof}
Fix a partition with $p$ blocks. Each block $P\in\cP$ defines a cut $(P,V\setminus P)$, and let $t_e$ be the number of these cuts crossed by $e$. A connected component of the quotient hypergraph of $e$ containing $r\ge2$ vertices corresponds to $r$ block cuts crossed by $e$ and contributes $r-1$ to $p-\kappa_e(\cP)$. Since $r\le2(r-1)$, we have $t_e\le2(p-\kappa_e(\cP))$. Thus,
\[
2D_G(\cP)\ge\sum_{e\in\cE}t_e=\sum_{P\in\cP}d_G(P)\ge p\lambda(G)\ge(p-1)\lambda(G).
\]
Dividing by $2(p-1)$ gives the fractional lower bound.

For the upper bound, take a minimum cut $(S,V\setminus S)$ and view its sides as the partition $\cP=\{S,V\setminus S\}$. Each crossing hedge contributes one to $D_G(\cP)$, and each remaining hedge contributes zero. Hence $D_G(\cP)=d_G(S)=\lambda(G)$. Since $|\cP|-1=1$,
\[
\frac{D_G(\cP)}{|\cP|-1}=D_G(\cP)=\lambda(G).
\]
Minimizing over all partitions gives $\FWPC(G)\le\lambda(G)$. Rounding down gives the integer bounds.
\end{proof}

The input model and definitions imply that the decision problem in Theorem~\ref{thm:main} belongs to NP. A certificate is a partition $\cP$, encoded by one block identifier per vertex using $O(|V|\log|V|)$ bits. A verifier checks that there are at least two blocks and computes $D_G(\cP)$ in polynomial time by finding the quotient components of each hedge. For an integer threshold $K$, the partition has integer score at most $K$ exactly when
\begin{equation}\label{eq:np-certificate}
D_G(\cP)<(K+1)(|\cP|-1).
\end{equation}
By the definition of WPC, an accepted certificate exists if and only if $\WPC(G)\le K$, establishing membership in NP.

\subsection{Full-support splits and their matrix representation}\label{subsec:matrix-representation}

If a hedge $e_i$ consists of exactly the two component hyperedges $T_i$ and $V\setminus T_i$, where $\varnothing\ne T_i\subsetneq V$, it divides $V$ into two nonempty parts. We call such a hedge a \emph{full-support split}, and a hedgegraph in which every hedge has this form a \emph{full-support split system}. Repeated splits are counted separately.

These instances admit a binary matrix representation. Let $A\in\{0,1\}^{m\times N}$, where $m\ge1$, $N\ge2$, and every row contains both a zero and a one. Write $[N]=\{1,\ldots,N\}$ and $[m]=\{1,\ldots,m\}$. Each column represents a vertex, and each row represents a hedge. The zero and one positions in row $i$ form its two component hyperedges:
\begin{equation}\label{eq:matrix-hedge}
e_i=\bigl\{\{v:A_{i,v}=0\},\{v:A_{i,v}=1\}\bigr\}.
\end{equation}
Denote the resulting hedgegraph by $G_A$. Each hedge has two nonempty component hyperedges covering all $N$ vertices. \emph{Complementing} a row exchanges all its zeros and ones, merely interchanging the two sides of its split and leaving the hedge unchanged. Copying a row creates a new hedge and its two component hyperedges; all copies are explicitly listed in the input. A \emph{row type} is a distinct full row vector, and its \emph{multiplicity} is the number of occurrences of that vector in the matrix.

For a set of row indices $R\subseteq[m]$, let $A_{R,v}$ be column $v$ restricted to the coordinates in $R$, called its \emph{projection} onto $R$. Distinct columns may become equal after the other rows are removed. Let $c_A(R)$ count the distinct projected columns. For nonempty $R$, define its density as $|R|/(c_A(R)-1)$, and let $\rho(A)$ be the maximum density:
\begin{equation}\label{eq:row-density}
c_A(R)=\bigl|\{A_{R,v}:v\in[N]\}\bigr|, \qquad \rho(A)=\max_{\varnothing\ne R\subseteq[m]}\frac{|R|}{c_A(R)-1}.
\end{equation}
Repeated rows are counted individually in $|R|$, whereas repeated columns contribute only one type to $c_A(R)$. We set $c_A(\varnothing)=1$, since all columns have the same empty projection. For nonempty $R$, every selected row is nonconstant, so $c_A(R)\ge2$ and the denominator is positive.

For a partition $\cP$, let $u_A(\cP)$ denote the number of rows that are constant within each block, counting repeated rows separately. Given $R\subseteq[m]$, put vertices with the same projected column $A_{R,v}$ into the same block, and denote the resulting partition by $\cP_R$. It has exactly $c_A(R)$ blocks, and every row in $R$ is constant within each block. A partition obtained by further splitting the blocks of $\cP_R$ is called a refinement of $\cP_R$. For example, if $R$ contains just one row, then $\cP_R$ is the two-block partition given by its zero and one sides.

\subsection{The matrix density formula}\label{subsec:matrix-density}

We now prove the matrix density formula.

\begin{lemma}\label{lem:partition-score}
For every partition $\cP$ with $p\ge2$ blocks,
\[
D_{G_A}(\cP)=m(p-1)-u_A(\cP).
\]
\end{lemma}
\begin{proof}
If a row is constant on every block of $\cP$, the quotient of its hedge has two connected components, contributing $p-2$. Otherwise, some block meets both sides of the split, so the quotient is connected and contributes $p-1$. Summing over all rows gives
\[
D_{G_A}(\cP)=u_A(\cP)(p-2)+(m-u_A(\cP))(p-1)=m(p-1)-u_A(\cP).
\]
\end{proof}

\begin{lemma}\label{lem:partition-density}
For every matrix considered here,
\[
\max_{\substack{\cP\text{ a partition of }[N]\\|\cP|\ge2}} \frac{u_A(\cP)}{|\cP|-1}=\rho(A).
\]
\end{lemma}
\begin{proof}
Given $\cP$, let $R$ contain all rows that are constant within each of its blocks, so $|R|=u_A(\cP)$. If $R=\varnothing$, the ratio for this partition is zero and is at most $\rho(A)$. Otherwise, columns in the same block have the same projection onto $R$, so every block of $\cP$ is contained in a block of $\cP_R$. Thus $\cP$ refines $\cP_R$, giving $|\cP|\ge|\cP_R|=c_A(R)$ and
\[
\frac{u_A(\cP)}{|\cP|-1}=\frac{|R|}{|\cP|-1}\le\frac{|R|}{c_A(R)-1}\le\rho(A).
\]
Conversely, for any nonempty $R$, the partition $\cP_R$ has $c_A(R)\ge2$ blocks, and every selected row is constant within each block. Hence $u_A(\cP_R)\ge|R|$. This inequality can be strict because additional rows outside $R$ may have the same property. Consequently,
\[
\frac{u_A(\cP_R)}{|\cP_R|-1}\ge\frac{|R|}{c_A(R)-1}.
\]
Maximizing over $R$ gives the reverse inequality.
\end{proof}

\begin{theorem}[Matrix density formula]\label{thm:matrix-density}
If $A$ has $m\ge1$ nonconstant rows, then
\begin{equation}\label{eq:matrix-density}
\FWPC(G_A)=m-\rho(A),\qquad \WPC(G_A)=m-\lceil\rho(A)\rceil.
\end{equation}
\end{theorem}
\begin{proof}
Lemma~\ref{lem:partition-score} writes each partition score as $m-u_A(\cP)/(|\cP|-1)$, and Lemma~\ref{lem:partition-density} identifies the maximum of the ratio. Therefore,
\[
\FWPC(G_A)=\min_{\cP}\left(m-\frac{u_A(\cP)}{|\cP|-1}\right)=m-\rho(A).
\]
The identity $\lfloor m-x\rfloor=m-\lceil x\rceil$ for integer $m$ gives the second formula.

Moreover, if $R^*$ attains $\rho(A)$, then Lemma~\ref{lem:partition-density} and its proof give
\[
\rho(A)=\frac{|R^*|}{c_A(R^*)-1}\le\frac{u_A(\cP_{R^*})}{|\cP_{R^*}|-1}\le\rho(A).
\]
Both inequalities are equalities. Since $|\cP_{R^*}|=c_A(R^*)$, it follows that $u_A(\cP_{R^*})=|R^*|$. Thus $\cP_{R^*}$ is simultaneously optimal for fractional and integer WPC.
\end{proof}

The density formulation replaces optimization over vertex partitions by optimization over row sets.

\section{NP-completeness and separation from hedge connectivity}
\label{sec:hardness}

We reduce Independent Set on subcubic graphs to row-density maximization.

\subsection{The reduction construction}
\label{subsec:construction}

Independent Set remains NP-complete on subcubic simple graphs, that is, graphs of maximum degree at most three~\cite{r31}. Let $(H,k)$ be an instance with $V(H)=[n]$ and $1\le k\le n$. Write $r=|E(H)|$ and let $\alpha(H)$ denote the maximum independent-set size of $H$. Adding a disjoint $K_2$ and three isolated vertices increases $\alpha(H)$ by four; increasing $k$ by four preserves the answer. We may therefore assume that $n\ge5$, $r\ge1$, and $\alpha(H)\ge3$.

Form an indexed family $\mathcal G$ containing one two-element group $\{u,v\}$ for every edge $uv\in E(H)$ and $3-d_H(v)$ private one-element groups $\{v\}$ for every vertex $v$. Groups with identical one-element contents have distinct indices. Let
\[
I=\{(g,v):g\in\mathcal G,\ v\in g\}.
\]
Each vertex belongs to exactly three groups, so
\[
|\mathcal G|=3n-r,\qquad |I|=3n.
\]
For $X\subseteq[n]$, let $c_H(X)$ count the groups meeting $X$. Counting vertex-group incidences gives
\begin{equation}\label{eq:group-coverage}
c_H(X)=3|X|-|E(H[X])|.
\end{equation}

Use columns indexed by
\[
z,\qquad e_g\ (g\in\mathcal G),\qquad p_{g,v}\ ((g,v)\in I),\qquad q_v\ (v\in[n]),\qquad q_0.
\]
For a binary row $a$, let $\operatorname{supp}(a)$ denote the set of columns in which its entry is one.

Define base rows $x_{g,v}$ for $(g,v)\in I$ and $b_j$ for $j\in\{0,1,\ldots,n\}$ by
\[
\operatorname{supp}(x_{g,v})=\{e_g,p_{g,v}\},\qquad\operatorname{supp}(b_j)=\{q_j\}.
\]
For each $(g,v)\in I$ and $j\in\{v,0\}$, define a reward row
\begin{equation}\label{eq:reward-row}
y_{g,v}^{j}=x_{g,v}\oplus b_j,
\end{equation}
where $\oplus$ denotes coordinatewise addition modulo two. Its support is
\[
\operatorname{supp}(y_{g,v}^{j})=\{e_g,p_{g,v},q_j\}.
\]
The group column $e_g$ creates an extra projection type when both incidences of an edge are selected, while the private columns $p_{g,v}$ distinguish individual incidences. The columns $q_v$ record the selected vertices. A reward row preserves the base projection partition precisely when both of its defining base types have been selected. The common marker $q_0$ provides one further available reward for each selected incidence at the cost of a single shared base type.

There are $4n+1$ base types and $6n$ reward types. The columns $p_{g,v}$ and $q_j$ distinguish all these types. Set
\begin{equation}\label{eq:multiplicities}
T=(4n+1)^2,\qquad B=6nT,\qquad W=B+1.
\end{equation}
Include $W$ copies of each base type and $T$ copies of each reward type. The total number of rows is
\begin{equation}\label{eq:row-count}
M=(4n+1)W+B.
\end{equation}
The inequality $W>B$ ensures that all reward rows together cannot compensate for one additional projection type. The factor $T$ separates the values associated with consecutive independence numbers after rounding, as needed for the integer-threshold reduction.

Finally, duplicate every column once, giving $N=2(7n-r+2)$ columns, and denote the resulting matrix by $A_H$. Every row is nonconstant because it is nonzero and vanishes at $z$. Column duplication preserves all projection-type counts and makes both sides of every split contain at least two vertices.

Let $L\subseteq I$ and $J\subseteq\{0,1,\ldots,n\}$ index the selected base types $x_{g,v}$ and $b_j$, respectively. Let $X_{L,J}$ contain one representative of each selected type, and put
\[
s=|L|+|J|,\qquad
h(L)=|\{g\in\mathcal G:|g|=2,\ (\{g\}\times g)\subseteq L\}|.
\]

\begin{lemma}\label{lem:base-projections}
For every such pair $(L,J)$,
\[
c_{A_H}(X_{L,J})=1+s+h(L).
\]
\end{lemma}
\begin{proof}
The column $z$ gives the zero pattern. For each $(g,v)\in L$, the projection of column $p_{g,v}$ has a single one in the coordinate corresponding to $x_{g,v}$. For each $j\in J$, column $q_j$ similarly gives a single one in the coordinate corresponding to $b_j$. These give $s$ distinct unit patterns.

A column $e_g$ gives zero, one of these unit patterns, or a pattern with two ones. The last case occurs exactly when both pairs $(g,u)$ and $(g,v)$ of a two-element group $g=\{u,v\}$ belong to $L$. Different such groups give different two-coordinate patterns. These are exactly the $h(L)$ additional types.
\end{proof}

\begin{lemma}\label{lem:reward-projections}
Adding $y_{g,v}^{j}$ to $X_{L,J}$ preserves its projection partition if and only if $(g,v)\in L$ and $j\in J$.
\end{lemma}
\begin{proof}
If both $x_{g,v}$ and $b_j$ are selected, their XOR is constant on every existing projection class. If $x_{g,v}$ is not selected, columns $p_{g,v}$ and $z$ have the same zero base projection but differ in $y_{g,v}^{j}$. If $b_j$ is not selected, the same argument applies to columns $q_j$ and $z$. Thus the projection partition is preserved exactly when both base types are selected.
\end{proof}

\begin{theorem}[Independent-set density]
\label{thm:independent-density}
The construction satisfies
\begin{equation}\label{eq:independent-density}
\rho(A_H)=W+\frac{6T\alpha(H)}{4\alpha(H)+1}.
\end{equation}
\end{theorem}
\begin{proof}
An optimal row set contains every copy of each selected type, since adding an identical copy preserves its projection classes and increases its numerator. A base type alone attains density $W$, whereas a row set without base types has density at most $B<W$.

Suppose that $s\ge1$ base types are selected. If $h(L)>0$, or a selected reward row is formed using a base type that has not been selected, Lemmas~\ref{lem:base-projections} and~\ref{lem:reward-projections} give a denominator of at least $s+1$. In either case the density is at most
\begin{equation}\label{eq:bad-density}
\frac{sW+B}{s+1}<W,
\end{equation}
so this cannot occur in an optimal set.

Thus, for each group $g$, $L$ contains at most one pair $(g,v)$. By Lemma~\ref{lem:reward-projections}, every reward row $y_{g,v}^{j}$ with $(g,v)\in L$ and $j\in J$ is constant on the projection classes of the selected base rows. Any missing such type can therefore be added without changing the denominator, while increasing the numerator. Hence an optimal set contains all these reward types and no others.

For such $(L,J)$, put $e(L,J)=|\{((g,v),j):(g,v)\in L,\ j\in J\cap\{v,0\}\}|$. There are $sW+Te(L,J)$ selected rows. Since $h(L)=0$ and the reward rows preserve the base projection partition, the denominator is $s$. Thus the density is
\[
\frac{sW+Te(L,J)}{s}=W+T\frac{e(L,J)}{|L|+|J|}.
\]
Consequently,
\begin{equation}\label{eq:independent-selection}
\rho(A_H)=W+T\beta(H),\qquad\beta(H)=\max_{(L,J)}\frac{e(L,J)}{|L|+|J|},
\end{equation}
where the maximum ranges over $L\subseteq I$ and $J\subseteq\{0,1,\ldots,n\}$ with $|L|+|J|>0$ and at most one pair $(g,v)\in L$ for each group $g$.

For any fixed vertex $v$, taking $J=\{v,0\}$ and $L=\{(g,v):g\in\mathcal G,\ v\in g\}$ gives a ratio of $6/5$, so $\beta(H)>1$. If $0\notin J$ or $J=\{0\}$, then $e(L,J)\le|L|$, giving a ratio at most one. Hence an optimum has $J=X\cup\{0\}$ with $\varnothing\ne X\subseteq[n]$. Also $\beta(H)<2$, since $e(L,J)\le2|L|$ and $e(L,J)>0$ requires $J\ne\varnothing$.

A selected pair $(g,v)\in L$ with $v\notin X$ contributes only one to $e(L,J)$; deleting it increases a ratio greater than one. If an unrepresented group $g$ meets $X$, choose $v\in g\cap X$. Adding $(g,v)$ contributes two to $e(L,J)$ and increases a ratio less than two. Thus, in an optimum, each group meeting $X$ supplies exactly one selected pair $(g,v)$ with $v\in X$, and no other pairs are selected.

Conversely, every nonempty $X\subseteq[n]$ admits such a choice: take $J=X\cup\{0\}$ and choose one pair $(g,v)$ with $v\in g\cap X$ for each group meeting $X$. This gives
\[
e(L,J)=2c_H(X),\qquad |L|+|J|=c_H(X)+|X|+1.
\]
Therefore,
\[
\beta(H)=\max_{\varnothing\ne X\subseteq[n]}\frac{2c_H(X)}{c_H(X)+|X|+1}.
\]
This expression is an increasing function of $c_H(X)/(|X|+1)$. By~\eqref{eq:group-coverage}, it remains to maximize
\[
\tau(H)=\max_{\varnothing\ne X\subseteq[n]}\frac{3|X|-|E(H[X])|}{|X|+1}.
\]
Since $\alpha(H)\ge3$, a maximum independent set gives $\tau(H)\ge\frac{3\alpha(H)}{\alpha(H)+1}>2$. If a maximizing $X$ contains an edge, choose $v\in X$ with internal degree $d=d_{H[X]}(v)\ge1$, and write $t=|X|$ and $c=c_H(X)$. By~\eqref{eq:group-coverage}, removing $v$ changes the coverage to $c-3+d$, and
\[
\frac{c-3+d}{t}-\frac{c}{t+1}=\frac{c-(3-d)(t+1)}{t(t+1)}>0,
\]
because $c/(t+1)>2\ge3-d$. The remaining set is nonempty, contradicting maximality. Therefore a maximizing $X$ is independent. Since $3|X|/(|X|+1)$ increases with $|X|$,
\[
\tau(H)=\frac{3\alpha(H)}{\alpha(H)+1},\qquad\beta(H)=\frac{6\alpha(H)}{4\alpha(H)+1}.
\]
Substitution in~\eqref{eq:independent-selection} proves the theorem.
\end{proof}

\subsection{Proof of NP-completeness}
\label{subsec:complexity-proof}

\begin{proof}[Proof of Theorem~\ref{thm:main}]
Theorems~\ref{thm:matrix-density} and~\ref{thm:independent-density} give
\begin{equation}\label{eq:hard-values}
\begin{aligned}
\WPC(G_{A_H})&=M-W-\left\lceil\frac{6T\alpha(H)}{4\alpha(H)+1}\right\rceil,\\
\FWPC(G_{A_H})&=M-W-\frac{6T\alpha(H)}{4\alpha(H)+1}.
\end{aligned}
\end{equation}
Put $\eta(a)=6Ta/(4a+1)$. Since $T=(4n+1)^2$, for $1\le a<n$ we have
\[
\eta(a+1)-\eta(a)=\frac{6T}{(4a+1)(4a+5)}>6.
\]
Thus $\lceil\eta(a)\rceil$ is strictly increasing for integer $1\le a\le n$. Setting
\[
K=M-W-\left\lceil\frac{6Tk}{4k+1}\right\rceil
\]
therefore gives
\[
\alpha(H)\ge k\quad\Longleftrightarrow\quad\WPC(G_{A_H})\le K.
\]
The threshold is nonnegative, since $\lceil\eta(k)\rceil\le B$ and $M-W-B=4nW\ge0$. For the fractional optimum $f=\FWPC(G_{A_H})$, put $d=M-W-f$. The second identity in~\eqref{eq:hard-values} gives
\[
\alpha(H)=\frac{d}{6T-4d},
\]
with positive denominator. Thus exact fractional computation also determines $\alpha(H)$ in polynomial time.

The construction has $M=O(n^4)$ rows and $N=O(n)$ columns, with every row copy listed explicitly. Its representation contains $MN=O(n^5)$ vertex incidences, and all identifiers and the threshold use $O(\log n)$ bits. Hence the reduction has polynomial output size.

The zero-side hyperedges of $b_0$ and any $b_v$ both contain the copies of $z$ and together cover all vertices. Hence $G_{A_H}$ is connected. After normalization at $z$ and column compression, all row supports have size at most three. Membership in NP follows from~\eqref{eq:np-certificate}.
\end{proof}

\subsection{Hedge connectivity}
\label{subsec:hedge-connectivity}

Let $\mu(A)$ be the maximum multiplicity of a split represented by the rows of $A$, treating a row and its complement as the same split.

\begin{proposition}\label{prop:split-cut}
For every matrix $A$ representing a full-support split system,
\begin{equation}\label{eq:split-connectivity}
\lambda(G_A)=m-\mu(A).
\end{equation}
\end{proposition}

\begin{proof}
A hedge avoids a nontrivial cut if and only if its split coincides with that cut, up to exchanging the two sides. Thus every cut is crossed by at least $m-\mu(A)$ hedges. A cut corresponding to a split of multiplicity $\mu(A)$ is crossed by exactly $m-\mu(A)$ hedges.
\end{proof}

The base and reward types are distinct and all vanish at $z$, so no two types are complements. Their split multiplicities therefore equal their row multiplicities. Since base types occur $W$ times and reward types occur
$T<W$ times, $\mu(A_H)=W$.
\begin{equation}\label{eq:separation}
\lambda(G_{A_H})=M-W.
\end{equation}
The WPC values in~\eqref{eq:hard-values} determine $\alpha(H)$, while the minimum cut value depends only on split multiplicities.

More generally, minimum hedge cut is computable in polynomial time when the number of components per hedge is bounded by a constant.

\begin{proposition}[Minimum hedge cut with bounded components]
\label{prop:bounded-component-cut}
If every hedge has at most $b$ connected components on $V$, a minimum hedge cut can be computed deterministically in time $2^{O(b)}\poly(\ell)$, where $\ell$ is the explicit input length.
\end{proposition}

\begin{proof}
A hedge with $q_i$ components avoids a cut exactly when each cut side is a union of those components. It therefore avoids precisely $2^{q_i-1}-1$ nontrivial unordered cuts. Enumerate these cuts for all hedges, add any fixed nontrivial cut, and evaluate the number of crossing hedges for every candidate. There are at most $m(2^{b-1}-1)+1$ candidates. If $\lambda(G)<m$, a minimum cut avoids at least one hedge and is enumerated. If $\lambda(G)=m$, the fixed cut is optimal. Each evaluation takes polynomial time.
\end{proof}

\subsection{Hardness on edge-colored simple graphs}
\label{subsec:simple-realization}

We next transfer hardness to edge-colored simple graphs. For $F=(U,L)$ with color classes $L_1,\ldots,L_m$, interpret each component of $(U,L_i)$, including an isolated vertex, as a component hyperedge of hedge $i$. All WPC and hedge connectivity notation for $F$ refers to this hedgegraph. Different colors may induce the same split.

The construction uses a decomposition of $K_{2m}$ into spanning double stars, obtained by adding the corresponding matching edge to each two-star forest in~\cite{r30}.

\begin{lemma}[Double-star decomposition]\label{lem:double-stars}

For every $m\ge1$, the complete graph on $\{a_1,b_1,\ldots,a_m,b_m\}$ decomposes into $m$ edge-disjoint spanning trees $T_1,\ldots,T_m$, where
\[
\begin{aligned}
E(T_i)={}&\{a_i b_i\}\cup\{a_i a_j,\,b_i b_j:j>i\}\cup\{a_i b_j,\,b_i a_j:j<i\}.
\end{aligned}
\]
\end{lemma}

\begin{proof}
In $T_i$, the centers $a_i,b_i$ are adjacent, and every other vertex is a leaf adjacent to one of them. Thus $T_i$ is a spanning tree.

For $i<j$, the edges $a_i a_j,b_i b_j$ belong to $T_i$, whereas $a_i b_j,b_i a_j$ belong to $T_j$. The edge $a_i b_i$ belongs only to $T_i$. These assignments cover every edge exactly once.
\end{proof}

\begin{lemma}[Simple-graph realization]\label{lem:simple-realization}
Let $A\in\{0,1\}^{m\times N}$, where $m\ge1$, $N\ge2$, and every row is nonconstant. Let $G_A$ be the full-support split system represented by $A$. There is a deterministic polynomial-time construction of a simple graph $F_A=(U,L)$ with an edge partition $L=L_1\mathbin{\dot\cup}\cdots\mathbin{\dot\cup}L_m$ such that $|U|=2mN$, $|L|=m(2mN-2)$, and each spanning subgraph $(U,L_i)$ has exactly two connected components, both nontrivial trees. Moreover,
\[
\FWPC(F_A)=\FWPC(G_A),\qquad \WPC(F_A)=\WPC(G_A),\qquad \lambda(F_A)=\lambda(G_A).
\]
In particular, $F_A$ is connected whenever $G_A$ is connected.
\end{lemma}
\begin{proof}
Replace each vertex $v\in[N]$ by
\[
C_v=\{a_{v,1},b_{v,1},\ldots,a_{v,m},b_{v,m}\},\qquad U=\bigcup_{v\in[N]}C_v.
\]
Apply Lemma~\ref{lem:double-stars} to each $C_v$. Let $T_{v,i}$ be the copy of $T_i$ obtained by replacing $a_j,b_j$ with $a_{v,j},b_{v,j}$, respectively. Assign all edges of $T_{v,i}$ to color $i$.

For each $i\in[m]$ and $\sigma\in\{0,1\}$, let
\[
V_{i,\sigma}=\{v\in[N]:A_{i,v}=\sigma\},\qquad U_{i,\sigma}=\bigcup_{v\in V_{i,\sigma}}C_v.
\]
Both sets $V_{i,0},V_{i,1}$ are nonempty because row $i$ is nonconstant.

Let $r_{i,\sigma}=\min V_{i,\sigma}$. For every $v\in V_{i,\sigma}\setminus\{r_{i,\sigma}\}$, add the edge $a_{r_{i,\sigma},i}a_{v,i}$ to color $i$. Let $L_i$ consist of the internal tree edges and these connecting edges, and set
\[
L=\bigcup_{i=1}^{m}L_i,\qquad F_A=(U,L).
\]
The internal tree edges are disjoint across colors by Lemma~\ref{lem:double-stars}. The connecting edges join different copy sets, so they cannot coincide with internal edges. They also cannot coincide across colors, since both endpoints of a connecting edge of color $i$ carry index $i$. Within each color and each side, every non-root copy set is joined to the root copy set exactly once. Thus $F_A$ is simple, and the color classes partition its edge set.

For each $i$ and $\sigma$, the connecting edges join the trees indexed by $V_{i,\sigma}$ into a single tree on $U_{i,\sigma}$. No edge of color $i$ joins $U_{i,0}$ to $U_{i,1}$. Hence $(U,L_i)$ consists of exactly these two trees. Each contains at least $2m\ge2$ vertices and is nontrivial. Consequently,
\[
|U|=2mN,\qquad |L_i|=2mN-2,\qquad|L|=m(2mN-2).
\]
The construction has encoding length $O(m^2N\log(mN))$ and takes deterministic polynomial time.

Let $\widehat A$ repeat each column of $A$ exactly $2m$ times, with the copies indexed by the corresponding vertices of $C_v$. The two component hyperedges of hedge $i$ in $G_{\widehat A}$ are $U_{i,0}$ and $U_{i,1}$.

Fix any partition $\cP$ of $U$ into at least two blocks. After contraction, each monochromatic tree remains connected on the quotient vertices that meet its vertex set. The two contracted trees of color $i$ are disjoint if no block meets both $U_{i,0}$ and $U_{i,1}$; otherwise, their union is connected. The same two cases apply to the component hyperedges of hedge $i$ in $G_{\widehat A}$. Therefore,
\[
\kappa_i^{F_A}(\cP)=\kappa_i^{G_{\widehat A}}(\cP)\quad\text{for every }i,
\]
and hence $D_{F_A}(\cP)=D_{G_{\widehat A}}(\cP)$. Minimizing the partition scores gives equality of both WPC values on $F_A$ and $G_{\widehat A}$. For two-block partitions, $D$ counts crossing hedges, so hedge connectivity agrees as well.

Repeating columns preserves every projected column-type count and preserves whether two rows are equal or complementary. Hence $\rho(\widehat A)=\rho(A)$ and $\mu(\widehat A)=\mu(A)$. By Theorem~\ref{thm:matrix-density} and Proposition~\ref{prop:split-cut}, both WPC values and hedge connectivity also agree on $G_{\widehat A}$ and $G_A$. This proves the three claimed equalities. Finally, if $G_A$ is connected, then $\lambda(F_A)=\lambda(G_A)>0$, so $F_A$ is connected.
\end{proof}

\begin{corollary}\label{cor:simple-hardness}
The integer-threshold decision problem for WPC is NP-complete even for connected unweighted simple graphs $F=(U,L)$ with a given edge partition into color classes $L_1,\ldots,L_m$, where each spanning subgraph $(U,L_i)$ has exactly two connected components, both nontrivial trees. On this class, exact computation of fractional WPC is NP-hard, whereas hedge connectivity is computable in deterministic polynomial time.
\end{corollary}
\begin{proof}

Apply Lemma~\ref{lem:simple-realization} to the matrices in Theorem~\ref{thm:main}, retaining the same decision threshold. The reduction has $M=O(n^4)$ rows and $N=O(n)$ columns. Its realization has $2MN=O(n^5)$ vertices and $M(2MN-2)=O(n^9)$ edges, with $O(n^9\log n)$ bits in an explicit encoding. Both values in~\eqref{eq:hard-values} are preserved.

For membership in NP, a partition $\cP$ into $p\ge2$ blocks is a certificate. Its quotient component counts are computable in polynomial time, and its integer score is at most $K$ exactly when $D_F(\cP)<(K+1)(p-1)$.

For any graph in this restricted class, the two components of each color class define a split of $U$. If the maximum multiplicity of a split is $\mu$, then $\lambda(F)=m-\mu$. The components and these multiplicities can be computed in polynomial time.
\end{proof}

\section{Complexity under restrictions on row supports}\label{sec:exact}

Section~\ref{sec:definitions} reduces WPC on full-support split systems to row-density maximization, but in general there are exponentially many row sets to compare. We identify support structures for which this optimization can be solved in polynomial time. The main obstacle is that distinct columns may acquire the same nonzero projection after rows are removed. Throughout this section, $A$ has at least one row and no constant rows.

\subsection{The linear intersection condition}\label{subsec:linear-supports}

Choose any column as a reference column. Complement each row that is one in this column, and leave the other rows unchanged. The reference column then becomes all zero. We call this operation \emph{normalization} with respect to the reference column. Complementing a row only exchanges the sides of its split, so it preserves $G_A$. It also preserves whether any two columns agree on any selected row set, and hence preserves every $c_A(R)$ and $\rho(A)$.

Group columns that agree on all rows into \emph{column classes}. We call the operation of retaining one representative from each class \emph{column compression}. Record the original vertices represented by each retained column. Identical columns remain equal after projection onto any row set, so compression preserves all projection-type counts. To recover a partition of the original vertex set, replace each representative by all vertices in its column class.

After normalization and column compression, the reference column belongs to the unique all-zero column class, denoted by $C_0$. For row $i$, let its \emph{support} $S_i$ be the set of column classes on which its normalized entries are one. Since the reference column is zero in every row, $C_0\notin S_i$. Thus $|S_i|$ counts column classes rather than original column copies.

For a fixed reference column, let
\[
s_{C_0}(A)=\max_{i\in[m]}|S_i|,
\]
and define
\[
s(A)=\min_{C_0}s_{C_0}(A),
\]
where the minimum ranges over all choices of the reference column. Thus, $s(A)\le k$ if and only if some reference column gives row supports of size at most $k$ after normalization and column compression.

For a fixed reference column, define the support hypergraph $Q$ whose vertices are the column classes other than $C_0$. For each row $i$, include one hyperedge copy with support $S_i$. For $X\subseteq V(Q)$, let $Q[X]$ have vertex set $X$ and all hyperedge copies whose supports are contained in $X$. Write
\[
\delta(Q)=\max_{\varnothing\ne X\subseteq V(Q)}\frac{|E(Q[X])|}{|X|},\qquad \mu(Q)=\max_T|\{i:S_i=T\}|.
\]
Here $\delta(Q)$ is the maximum induced density, with hyperedge copies counted separately, and $\mu(Q)$ is the maximum support multiplicity. After normalization, two rows represent the same split if and only if their supports coincide. Hence $\mu(Q)=\mu(A)$.

The supports satisfy the \emph{linear intersection condition} if any two distinct support sets have at most one common vertex. Repeated copies of the same support are allowed; only distinct support sets are compared. Under this condition, the matrix density is determined by the maximum induced density and maximum support multiplicity of $Q$.

\begin{theorem}\label{thm:linear-support}
Suppose that the normalized row supports satisfy the linear intersection condition for a reference column, and let $Q$ be the corresponding support hypergraph. Then
\[
\rho(A)=\max\{\delta(Q),\mu(Q)\}.
\]
On this class, both WPC values and an optimal partition can be computed in deterministic polynomial time.
\end{theorem}
\begin{proof}
Taking all copies of a support of maximum multiplicity gives a row set with exactly two projection types. Hence $\rho(A)\ge\mu(Q)$.

Let $X$ attain $\delta(Q)$ and put $R=\{i:S_i\subseteq X\}$. Since $Q$ has a nonempty hyperedge, $\delta(Q)>0$ and $R\ne\varnothing$. We have $|R|=|E(Q[X])|$, and all columns outside $X$ project to zero. Thus
\[
\rho(A)\ge\frac{|R|}{c_A(R)-1}\ge\frac{|E(Q[X])|}{|X|}=\delta(Q).
\]
Together, these choices prove
\[
\rho(A)\ge\max\{\delta(Q),\mu(Q)\}.
\]

For the upper bound, if $\rho(A)\le\mu(Q)$, there is nothing to prove. It remains to consider $\rho(A)>\mu(Q)$ and establish $\rho(A)\le\delta(Q)$. Choose a nonempty row set $R$ attaining $\rho(A)$. Adding the remaining copies of a selected row type does not change the number of projection types and increases the numerator. Thus $R$ contains all copies of every selected type.

Suppose that two distinct vertices $u,v$ appearing in at least one selected support have the same projection. Every selected support then contains either both vertices or neither. Their common projection is nonzero, so some selected support contains both. If two distinct selected support sets contained them, their intersection would contain at least two vertices, violating the linear intersection condition. Hence all selected rows whose supports contain $u$ or $v$ have one common support $T$.

Let $w$ be the number of copies of $T$, so $w\le\mu(Q)$. Delete these rows from $R$ to obtain $R'$, with $|R'|=|R|-w$. The columns of $u,v$ become zero on the remaining rows, so their former common nonzero type merges into the zero type. Removing rows cannot distinguish previously equal projections; therefore the number of nonzero types decreases by at least one, and $c_A(R')<c_A(R)$.

Moreover, $R'$ is nonempty. Otherwise, every row of $R$ would have support $T$, producing just the all-zero and all-one projection types and density $w\le\mu(Q)$, contrary to $\rho(A)>\mu(Q)$. Since $R'$ is nonempty and no row is constant, $c_A(R)>c_A(R')\ge2$. Thus,
\[
\frac{|R'|}{c_A(R')-1}=\frac{|R|-w}{c_A(R')-1}\ge\frac{|R|-w}{(c_A(R)-1)-1}>\frac{|R|}{c_A(R)-1}.
\]
The strict inequality uses $|R|>w(c_A(R)-1)$, which follows from $\rho(A)>\mu(Q)\ge w$. This contradicts the optimality of $R$.

It follows that all vertices incident with selected supports have distinct nonzero projections. Let $X$ be their set. Columns outside $X$ project to zero, and each vertex in $X$ contributes one distinct nonzero type. Therefore, $c_A(R)-1=|X|$.

Every selected support lies entirely in $X$, so all corresponding hyperedge copies belong to $E(Q[X])$. Hence $|R|\le|E(Q[X])|$, and
\[
\rho(A)=\frac{|R|}{c_A(R)-1}=\frac{|R|}{|X|}\le\frac{|E(Q[X])|}{|X|}\le\delta(Q).
\]
Combining the bounds gives
\[
\rho(A)=\max\{\delta(Q),\mu(Q)\}.
\]
\end{proof}

The proof of Theorem~\ref{thm:linear-support} also shows how to choose a row set attaining $\rho(A)$. If $\mu(Q)\ge\delta(Q)$, choose a support of maximum multiplicity and select all rows having that support. Otherwise, choose a set $X$ attaining $\delta(Q)$ and select all rows whose supports are contained in $X$. Group columns with identical projections onto the selected rows, then replace each compressed column by the original columns it represents. By Theorem~\ref{thm:matrix-density}, the resulting partition is optimal for both integer and fractional WPC.

Exact polynomial-time algorithms for maximum induced density in weighted hypergraphs are known~\cite{r27}. We compute $\delta(Q)$ and a nonempty maximizing vertex set by reducing maximum-weight closure to minimum cut~\cite{r21}. For a nonnegative rational threshold $a=u/v$, where $u,v$ are integers and $v>0$,
\[
\delta(Q)>a\quad\Longleftrightarrow\quad\max_{X\subseteq V(Q)}\bigl(v|E(Q[X])|-u|X|\bigr)>0.
\]
The right-hand side can be computed using the standard reduction from maximum-weight closure to minimum cut. Give each hyperedge-copy node weight $v$ and each vertex node weight $-u$, and require the selection of a hyperedge node to imply the selection of all its vertex nodes. For a fixed vertex set $X$, choosing every induced hyperedge copy is optimal, and the resulting weight is exactly $v|E(Q[X])|-u|X|$.

More explicitly, add a source and a sink. Add an arc of capacity $v$ from the source to each hyperedge node, an arc of capacity $u$ from each vertex node to the sink, and an arc of capacity $mv+1$ from each hyperedge node to each of its vertex nodes, where $m$ is the number of hyperedge copies. The cut whose source side contains only the source has capacity $mv$. A minimum cut therefore cannot cross an arc of capacity $mv+1$, ensuring that a hyperedge node on the source side has all its vertex nodes on that side as well. Such a cut has capacity $mv$ minus the corresponding closure weight. Hence one minimum-cut computation performs the threshold test.

Let $q=|V(Q)|$. The maximum induced density belongs to the candidate set
\[
\left\{\frac{j}{r}:0\le j\le m,\ 1\le r\le q\right\}.
\]
Sorting the distinct candidates and using binary search therefore determines $\delta(Q)$ exactly. Two distinct candidates differ by at least $1/q^2$. Since $Q$ contains at least one nonempty hyperedge, $\delta(Q)\ge1/q$, and hence $\delta(Q)-1/(2q^2)>0$. Running the closure subroutine again at threshold $a=\delta(Q)-\frac{1}{2q^2}$ returns a closure of positive weight and hence a nonempty vertex set of density greater than $a$. By the separation between candidate values, this density must equal $\delta(Q)$. Thus the subroutine returns both the exact optimum and a set attaining it. We denote it by $\mathrm{MaxDensity}(Q)$. Algorithm~\ref{alg:linear-wpc} gives the full WPC algorithm under the linear intersection condition.

The main algorithm checks at most $N$ reference columns. Normalization and verification of the linear intersection condition take polynomial time for each column. The density subroutine has at most $(m+1)q$ candidate values and makes $O(\log((m+1)q))$ minimum-cut computations. Each network and the bit lengths of its capacities have polynomial size, so the entire algorithm runs in deterministic polynomial time.

\begin{algorithm}[!htb]
\caption{Exact WPC computation under the linear intersection condition}
\label{alg:linear-wpc}
\small
\begin{algorithmic}[1]
\Require $A\in\{0,1\}^{m\times N}$, with $m\ge1$, $N\ge2$, and no constant rows
\Ensure Both WPC values and a common optimal partition, or a report that no reference column satisfies the condition
\State Merge identical columns of $A$ to obtain $B$, recording the original vertices represented by each column
\State $\mathrm{found}\gets\mathrm{false}$
\For{each column $j$ of $B$}
    \State $\widehat B\gets B$
    \State Complement each row of $\widehat B$ whose entry in column $j$ is one
    \State Let $C_0$ be column $j$ and $S_i$ the set of one positions in row $i$
    \If{any two distinct supports intersect in at most one vertex}
        \State $\mathrm{found}\gets\mathrm{true}$
        \State \textbf{break}
    \EndIf
\EndFor
\If{$\mathrm{found}=\mathrm{false}$}
    \State \Return No suitable reference column; the algorithm does not apply
\EndIf
\State Construct $Q$ on the columns other than $C_0$, with one hyperedge $S_i$ per row, retaining copies
\State Count support multiplicities and choose a support $T$ of maximum multiplicity
\State $\mu\gets|\{i:S_i=T\}|$
\State $(\delta,X)\gets\Call{MaxDensity}{Q}$
\If{$\mu\ge\delta$}
    \State $R\gets\{i:S_i=T\}$
\Else
    \State $R\gets\{i:S_i\subseteq X\}$
\EndIf
\State Group columns of $\widehat B$ with equal projections onto $R$ into the same block
\State Replace each representative by all its original vertices to obtain $\cP_R$
\State $\rho\gets\max\{\mu,\delta\}$
\State \Return $\FWPC(G_A)=m-\rho$, $\WPC(G_A)=m-\lceil\rho\rceil$, and $\cP_R$
\end{algorithmic}
\end{algorithm}

\subsection{Bounded row supports}\label{subsec:small-support}

Assume $s(A)\le2$. Choose a reference column such that every row support has size at most two after normalization and column compression. Fix this reference column, and construct a multigraph $Q_A$ on the column classes other than the all-zero class $C_0$.

A row with support $S_i=\{u\}$ becomes a loop at $u$, and a row with support $S_i=\{u,v\}$ becomes an edge joining $u$ and $v$. There are no empty supports because no row is constant. Each row copy contributes a separate edge, so parallel edges and parallel loops are allowed.

For $X\subseteq V(Q_A)$, the induced subgraph $Q_A[X]$ includes edges with both endpoints in $X$ and loops at vertices of $X$. Each edge copy is counted separately, and each loop copy is counted once. Write
\[
\delta(Q_A)=\max_{\varnothing\ne X\subseteq V(Q_A)}\frac{|E(Q_A[X])|}{|X|}.
\]
Let $\beta(Q_A)$ be the maximum number of edges joining the same pair of distinct vertices, or zero if there are no nonloop edges.

\begin{corollary}[Row supports of size at most two]\label{cor:support-two}
If $s(A)\le2$, then
\begin{equation}\label{eq:support-two}
\rho(A)=\max\{\delta(Q_A),\beta(Q_A)\}.
\end{equation}
On this class, fractional WPC, integer WPC, and a common optimal partition can be computed in deterministic polynomial time.
\end{corollary}
\begin{proof}
Any two distinct supports of size at most two intersect in at most one vertex, so the linear intersection condition holds. By Theorem~\ref{thm:linear-support}, $\rho(A)=\max\{\delta(Q_A),\mu(Q_A)\}$. The multiplicity of a singleton support $\{u\}$ is the number of loops at $u$. Taking $X=\{u\}$ includes all these loops in $Q_A[X]$, so this multiplicity is at most $\delta(Q_A)$. The multiplicity of a two-element support $\{u,v\}$ is the number of edges joining $u$ and $v$, whose maximum is $\beta(Q_A)$. Hence $\max\{\delta(Q_A),\mu(Q_A)\}=\max\{\delta(Q_A),\beta(Q_A)\}$. The algorithmic conclusion follows from the exact algorithm under the linear intersection condition.
\end{proof}

Increasing the bound on $s(A)$ from two to three
already yields NP-hardness.

\begin{corollary}[Hardness with bounded row supports]
\label{cor:bounded-support-hardness}
The integer-threshold decision problem for WPC is NP-complete even for connected unweighted full-support split systems represented by matrices $A$ with $s(A)=3$. Exact computation of fractional WPC is NP-hard on the same class.
\end{corollary}
\begin{proof}
Use the zero marker $z$ as the reference column in the construction of Section~\ref{subsec:construction}. Base supports have size at most two and reward supports have size three before column duplication. Compression cannot increase these sizes, so $s(A_H)\le3$.

For a singleton group $g=\{v\}$, columns $e_g$ and $p_{g,v}$ are identical. Project onto one representative of each base type. The columns $p_{g,v}$ and $q_j$ give distinct unit patterns, whereas for a two-element group $g=\{u,v\}$, column $e_g$ has ones precisely in the coordinates corresponding to $x_{g,u}$ and $x_{g,v}$. Together with the zero column $z$, these give $N_c=4n+r+2$ column classes after compression. Since $r\ge1$, a reward row associated with a two-element group has compressed support size exactly three.

Every nonzero column class is one in some base row. Choosing it as reference complements that row, whose compressed support with reference $z$ has size at most two; its new support therefore has size at least $N_c-2>3$. Row complementation preserves equality of columns, so these class counts do not depend on the reference. Thus $s(A_H)=3$. The hardness and membership claims follow from Theorem~\ref{thm:main} and~\eqref{eq:hard-values}.
\end{proof}

\section{Approximation schemes with bounded components}
\label{sec:approximation}

We now give approximation schemes on a class containing all full-support split systems. Throughout this section, every hedge has at most $b\ge2$ connected components on the entire vertex set. We treat $b$ as fixed. Let $\ell$ denote the explicit input length and $N=|V|$. For a partition with at least two blocks, write
\[
F_G(\cP)=\frac{D_G(\cP)}{|\cP|-1},\qquad I_G(\cP)=\lfloor F_G(\cP)\rfloor.
\]
Their optimal values are $\FWPC(G)$ and $\WPC(G)$, respectively.

\subsection{Efficient approximation schemes}\label{subsec:ptas}

Let $\pi_i$ be the component partition of hedge $i$, with $q_i\le b$ blocks. A partition $\mathcal Q$ is coarser than $\cP$ if every block of $\cP$ is contained in a block of $\mathcal Q$. The join of two partitions is the finest partition coarser than both. Contracting by $\cP$ merges precisely the components connected by its blocks, so
\[
\kappa_i(\cP)=|\cP\vee\pi_i|.
\]
Define
\begin{equation}\label{eq:general-density}
\sigma_G(\cP)=\sum_{i=1}^m\bigl(\kappa_i(\cP)-1\bigr),\qquad\rho_G=\max_{|\cP|\ge2}\frac{\sigma_G(\cP)}{|\cP|-1}.
\end{equation}
Rearranging~\eqref{eq:partition-contribution} gives
\begin{equation}\label{eq:general-density-values}
F_G(\cP)=m-\frac{\sigma_G(\cP)}{|\cP|-1},\qquad\FWPC(G)=m-\rho_G,\quad\WPC(G)=m-\lceil\rho_G\rceil.
\end{equation}
For full-support split systems, $\sigma_G(\cP)=u_A(\cP)$ and $\rho_G=\rho(A)$.

A \emph{binary coarsening} of $\pi_i$ is a partition of $V$ into two nonempty unions of its components. There are $2^{q_i-1}-1\le C_b:=2^{b-1}-1$ such unordered partitions.

\begin{lemma}\label{lem:small-generating-set}
If a family of binary partitions has common refinement $\cP$ with $p\ge2$ blocks, some nonempty subfamily of at most $p-1$ partitions has the same common refinement.
\end{lemma}
\begin{proof}
Start with the one-block partition. Until it equals $\cP$, some member of the family separates two vertices in the same current block but different blocks of $\cP$. Refining by that member increases the number of blocks. At most $p-1$ such refinements suffice.
\end{proof}

\begin{lemma}\label{lem:optimal-coarsening}
Let $\cP^*$ attain $\rho_G$, put $t=|\cP^*|-1$, and let $\mathcal Q$ be a coarsening of $\cP^*$ with $q$ blocks. Then
\begin{equation}\label{eq:coarsening-progress}
\sum_i\bigl(\kappa_i(\cP^*)-\kappa_i(\mathcal Q)\bigr)\ge\rho_G(t-q+1).
\end{equation}
In particular, any two blocks of $\cP^*$ belong to different blocks of $\cP^*\vee\pi_i$ for at least $\rho_G$ hedges $i$.
\end{lemma}
\begin{proof}
For $q\ge2$, the definition of $\rho_G$ gives $\sigma_G(\mathcal Q)\le\rho_G(q-1)$, while $\sigma_G(\cP^*)=\rho_Gt$. Their difference proves \eqref{eq:coarsening-progress}. If $q=1$, then $\sigma_G(\mathcal Q)=0$, giving the same result.

For the last assertion, merge the two chosen blocks of $\cP^*$. A hedge's quotient component count decreases by one exactly when those blocks belong to different quotient components, and otherwise does not change. For this coarsening, $q=t$, so \eqref{eq:coarsening-progress} gives a total decrease of at least $\rho_G$.
\end{proof}

The pair-separation bound in Lemma~\ref{lem:optimal-coarsening} reduces deterministic recovery to hitting a bounded collection of dense sets of hedges.

\begin{lemma}[Deterministic sampling]
\label{lem:deterministic-sampling}
Let $m\ge1$, $a\ge2$, and $K\ge1$ be integers. One can deterministically construct a family $\mathcal F$ of subsets of $[m]$ in time
\[
(m+a)^{O(1)}(K+1)^{O(1)}
\]
such that every $S\in\mathcal F$ has size $O(a\log(K+1))$ and the following property holds: for any sets $E_1,\ldots,E_h\subseteq[m]$ satisfying
\[
h\le\binom{K+1}{2},\qquad |E_j|\ge m/a,
\]
some $S\in\mathcal F$ intersects every $E_j$. All polynomial exponents are constants.
\end{lemma}

\begin{proof}
Choose a prime $q_0$ with
\[
\max\{2m,8a\}\le q_0<2\max\{2m,8a\}.
\]
Assign $k_0=\lfloor q_0/m\rfloor$ distinct elements of $\mathbb F_{q_0}$ to each index in $[m]$. Any remaining field elements are left unassigned and treated as dummy labels.

A seed $(s,d)\in\mathbb F_{q_0}^2$ generates the batch
\[
d,\ s+d,\ 2s+d,\ldots,(8a-1)s+d.
\]
For a uniform seed, these labels are uniform and pairwise independent, since their coefficients of $s$ are distinct modulo $q_0$.

Fix any $E\subseteq[m]$ with $|E|\ge m/a$. The labels assigned to $E$ occupy a fraction at least
\[
\frac{k_0|E|}{q_0}\ge\frac{k_0m}{aq_0}\ge\frac1{2a}
\]
of the field. Let $Z_E$ count the labels in a batch that are assigned to $E$. Pairwise independence gives
\[
\mathbb E Z_E\ge4,\qquad\operatorname{Var}(Z_E)\le\mathbb E Z_E.
\]
By Chebyshev's inequality,
\[
\Pr(Z_E=0)\le\frac{\operatorname{Var}(Z_E)}{(\mathbb E Z_E)^2}\le\frac14.
\]
Thus the bad-seed set
\[
\mathcal B_E=\{(s,d)\in\mathbb F_{q_0}^2:Z_E(s,d)=0\}
\]
has density at most $1/4$.

We use an auxiliary expander graph whose vertices are the seeds in $\mathbb F_{q_0}^2$. Each vertex therefore represents one batch of $8a$ labels. Choose an explicit undirected $\Delta$-regular graph $\Gamma$ on this seed space, allowing loops and parallel edges, whose transition matrix $T$ satisfies
\[
\|T-J\|_2\le\frac14,\qquad J=\frac1{q_0^2}\mathbf1\mathbf1^{\mathsf T}.
\]
Here $\Delta$ is a constant and $\|\cdot\|_2$ is the Euclidean operator norm. Such a graph is obtained from the explicit Margulis--Gabber--Galil construction~\cite{r32} by making the walk lazy and taking a suitable fixed power of its transition matrix.

Set $g=\left\lceil2\log_2(K+1)\right\rceil+2$. For each walk $W$ in $\Gamma$ visiting $g$ vertices, collect the batches generated by its vertices. Discard dummy labels and repeated indices, and let $S(W)\subseteq[m]$ be the resulting set. Define
\[
\mathcal F=\{S(W):W\text{ visits }g\text{ vertices in }\Gamma\}.
\]
Every member has size at most $8ag=O(a\log(K+1))$.

To verify the required hitting property, consider a random walk $U_1,\ldots,U_g$ starting at a uniform vertex and choosing each successive edge uniformly. List all $q_0^2$ seeds in a fixed order. For any seed set $\mathcal B$ with $\beta=|\mathcal B|/q_0^2\le1/4$, let $D_{\mathcal B}$ be the diagonal matrix whose $j$th diagonal entry is $1$ if the $j$th seed belongs to $\mathcal B$, and $0$ otherwise. Write $\mathbf1_{\mathcal B}$ for the indicator vector of $\mathcal B$ in this ordering. Since
\[
D_{\mathcal B}JD_{\mathcal B}=\frac1{q_0^2}\mathbf1_{\mathcal B}\mathbf1_{\mathcal B}^{\mathsf T},
\]
we have
\[
\begin{aligned}
\|D_{\mathcal B}TD_{\mathcal B}\|_2
&\le\|D_{\mathcal B}JD_{\mathcal B}\|_2+\|D_{\mathcal B}(T-J)D_{\mathcal B}\|_2\\
&\le \beta+\frac14.
\end{aligned}
\]
Consequently,
\[
\begin{aligned}
\Pr(U_j\in\mathcal B\text{ for all }j)
&=\frac1{q_0^2}\mathbf1_{\mathcal B}^{\mathsf T}(D_{\mathcal B}TD_{\mathcal B})^{g-1} \mathbf1_{\mathcal B}\\
&\le
\beta(\beta+1/4)^{g-1}\\
&\le 2^{-g-1}.
\end{aligned}
\]

Now fix any collection $E_1,\ldots,E_h$ as in the statement. The set $S(W)$ misses $E_j$ exactly when every vertex of $W$ belongs to $\mathcal B_{E_j}$. A union bound therefore gives
\[
\Pr\bigl(S(W)\cap E_j=\varnothing\text{ for some }j\bigr)\le h\,2^{-g-1}\le\binom{K+1}{2}2^{-g-1}<1.
\]
Hence some walk produces a set intersecting every $E_j$.

The family is constructed by enumerating all $q_0^2\Delta^{g-1}$ walks. Since
\[
q_0=O(m+a), \qquad \Delta^{g-1}=(K+1)^{O(1)},
\]
the number of walks is polynomial in $m+a$ and $K+1$, with absolute exponents. Prime search, graph-neighbor computation, and generation of the batches and their index sets also take polynomial time in these quantities. This proves the construction bound.
\end{proof}

\begin{lemma}[Recovery at high density]
\label{lem:dense-recovery}
Let $G$ be a hedgegraph with $m$ hedges, each having at most $b$ connected components, and let $\ell$ be its explicit input length. For every integer $a\ge2$, there are deterministic and randomized algorithms with running times
\[
2^{O_b(a\log\log(a+3))}\poly(\ell)\quad\text{and}\quad2^{O_b(a)}\poly(\ell),
\]
respectively. Both algorithms return a partition with at least two blocks. If $\rho_G\ge m/a$, the deterministic algorithm returns a common optimum for integer and fractional WPC, and the randomized algorithm does so with probability at least $2/3$.
\end{lemma}

\begin{proof}
If every hedge is connected on $V$, both scores equal $m$ for every nontrivial partition, so return any such partition. Otherwise, fix a binary coarsening $\cP_0$ of the component partition of a disconnected hedge. Its two-block score is at most $m-1$.

Both algorithms retain $\cP_0$ as a candidate, evaluate all candidates in the original hedgegraph, and return a candidate minimizing $F_G$. Since $I_G(\cP)=\lfloor F_G(\cP)\rfloor$, this also minimizes the integer score among the candidates.

Put $K=a(b-1)$. For the correctness analysis, assume $\rho_G\ge m/a$ and fix an optimal partition $\cP^*$ with $t+1$ blocks. Then
\[
\rho_Gt=\sigma_G(\cP^*)\le m(b-1),\qquad t\le K.
\]
Write $\pi_i^*=\cP^*\vee\pi_i$. Every binary coarsening of $\pi_i^*$ is both a binary coarsening of $\pi_i$ and a coarsening of $\cP^*$.

The randomized algorithm chooses a power of two $H$ with $2m\le H<4m$. Assign $\lfloor H/m\rfloor$ distinct labels to each hedge, leaving the remaining labels dummy. A uniform label samples each hedge with probability at least $1/(2m)$, using exactly $\log_2 H$ random bits.

In each trial, draw $K$ independent labels. For each sampled hedge, enumerate the options of skipping it or selecting one of its binary coarsenings. A dummy label only permits skipping. For every combination selecting at least one binary partition, include its common refinement as a candidate. There are at most $(C_b+1)^K$ combinations per trial.

To analyze a trial, start with the one-block partition and follow choices that are coarsenings of $\cP^*$. Suppose the current partition $\mathcal Q$ has $q<t+1$ blocks, and put $r=t-q+1$. By \eqref{eq:coarsening-progress}, $\sum_i\bigl(\kappa_i(\cP^*)-\kappa_i(\mathcal Q)\bigr)\ge \rho_Gr$. Each summand lies between zero and $b-1$, so at least $\rho_Gr/(b-1)$ hedges satisfy $\kappa_i(\cP^*)>\kappa_i(\mathcal Q)$. 

For such a hedge, some block of $\mathcal Q$ meets two different blocks of $\pi_i^*$. A binary coarsening of $\pi_i^*$ can separate them. Refining $\mathcal Q$ by that binary partition increases its number of blocks while keeping it a coarsening of $\cP^*$.

The probability of sampling such a hedge is at least
\[
\frac{\rho_Gr}{2m(b-1)}\ge\frac r{2K}.
\]
For the analysis, select a progressing binary coarsening whenever such a hedge is sampled, and otherwise skip the draw. These choices describe a branch included in the enumeration.

Induction on $r$ shows that this branch reaches $\cP^*$ within $r$ further draws with probability at least $r!/(2K)^r$. Indeed, a progressing draw reduces $r$ by at least one, and $j!/(2K)^j$ is nonincreasing for $0\le j\le K$. After reaching $\cP^*$, all remaining draws can be skipped. Thus a trial includes $\cP^*$ with probability at least
\[
\frac{t!}{(2K)^t}\ge\frac{K!}{(2K)^K}\ge(2e)^{-K}>6^{-K}.
\]
Repeating $2\cdot6^K$ independent trials makes the failure probability at most $e^{-2}<1/3$. For fixed $b$, the sampling, enumeration, and evaluations take $2^{O_b(a)}\poly(\ell)$ time.

The deterministic algorithm constructs the family $\mathcal F$ from Lemma~\ref{lem:deterministic-sampling} with parameters $m,a,K$. For each $S\in\mathcal F$, generate all binary coarsenings of $\pi_i$ for $i\in S$. For every nonempty subfamily of at most $K$ generated binary partitions, include its common refinement as a candidate.

To prove correctness, for distinct blocks $B,C$ of $\cP^*$ define
\[
E_{BC}=\{i:\text{$B$ and $C$ lie in different blocks of $\pi_i^*$}\}.
\]
Lemma~\ref{lem:optimal-coarsening} gives $|E_{BC}|\ge\rho_G\ge m/a$. There are at most $\binom{K+1}{2}$ such pairs. Hence Lemma~\ref{lem:deterministic-sampling} guarantees a set $S\in\mathcal F$ intersecting every $E_{BC}$.

For this $S$, each pair $B,C$ can be separated by a binary coarsening of some $\pi_i^*$ with $i\in S$. Every such binary partition is generated by the algorithm and is a coarsening of $\cP^*$. These partitions separate all pairs of blocks of $\cP^*$, so their common refinement is exactly $\cP^*$. Lemma~\ref{lem:small-generating-set} supplies a subfamily of at most $t\le K$ with the same refinement. Thus $\cP^*$ is among the candidates, and the returned partition is a common optimum.

For the running time, each $S\in\mathcal F$ generates at most $L=|S|C_b=O(aC_b\log(K+1))$ binary partitions. The number of subfamilies considered for this $S$ is at most
\[
\sum_{j=1}^{\min(K,L)}\binom Lj\le\left[e\left(1+\frac LK\right)\right]^K=2^{O_b(a\log\log(a+3))}.
\]
The sampling family has at most $(m+a)^{O(1)}(K+1)^{O(1)}$ members. Combining these bounds with the construction time of $\mathcal F$ and polynomial-time partition refinement and score evaluation gives the stated deterministic running time.
\end{proof}

\begin{theorem}\label{thm:partition-ptas}
For every fixed integer $b\ge2$ and $0<\varepsilon\le1$, there are deterministic and randomized algorithms on unweighted hedgegraphs in which every hedge has at most $b$ connected components on the entire vertex set, including isolated vertices. Each algorithm returns a partition $\cP$ with at least two blocks. The bounds
\begin{align*}
F_G(\cP)&\le(1+\varepsilon)\FWPC(G),\\
I_G(\cP)&\le(1+\varepsilon)\WPC(G)
\end{align*}
hold for the deterministic algorithm and hold simultaneously with probability at least $2/3$ for the randomized algorithm. Their respective running times are
\[
2^{O_b(\varepsilon^{-1}\log\log(1/\varepsilon+3))}\poly(\ell)\quad\text{and}\quad 2^{O_b(\varepsilon^{-1})}\poly(\ell).
\]
The polynomial degrees are independent of $b$ and $\varepsilon$.
\end{theorem}
\begin{proof}
If $m=0$ or every hedge is connected, return any nontrivial partition, which is optimal for both objectives. Otherwise put $a=\lceil1/\varepsilon\rceil+1$ and run the corresponding algorithm of Lemma~\ref{lem:dense-recovery}. Its fixed two-block candidate ensures that the returned partition $\widehat{\cP}$ has both scores at most $m-1$. If $m=1$, both scores are zero and hence optimal. Assume $m\ge2$.

If $\rho_G\ge m/a$, the algorithm recovers a common exact optimum, with its stated success probability. Otherwise, \eqref{eq:general-density-values} gives
\[
\FWPC(G)=m-\rho_G\ge\frac{m(a-1)}a.
\]
Since $m$ is an integer,
\[
\left\lceil\frac ma\right\rceil=1+\left\lfloor\frac{m-1}a\right\rfloor \le1+\frac{m-1}a,
\]
and hence
\[
\WPC(G)=m-\lceil\rho_G\rceil \ge m-\left\lceil\frac ma\right\rceil \ge\frac{(m-1)(a-1)}a.
\]
Since both returned scores are at most $m-1$, the two approximation ratios are at most $a/(a-1)\le1+\varepsilon$. Substituting $a=O(1/\varepsilon)$ gives the stated running times.
\end{proof}

\begin{corollary}[Full-support split systems]\label{cor:split-eptas}
On full-support split systems, the deterministic and randomized running times in Theorem~\ref{thm:partition-ptas} are
\[
2^{O(\varepsilon^{-1}\log\log(1/\varepsilon+3))} \poly(m,N) \quad\text{and}\quad 2^{O(\varepsilon^{-1})}\poly(m,N),
\]
respectively. Each algorithm returns a partition satisfying both approximation bounds, with the success guarantee stated in that theorem.
\end{corollary}
\begin{proof}
Every hedge has exactly two components, so take $b=2$. The explicit matrix and hedgegraph input sizes are polynomially equivalent.
\end{proof}

For edge-colored graphs, replacing each monochromatic component by a hyperedge preserves every quotient component count and every partition score. Consequently Theorem~\ref{thm:partition-ptas} applies whenever every color has at most $b$ components on the full vertex set. In particular, taking $b=2$ covers the two-tree color classes in Corollary~\ref{cor:simple-hardness}.

\subsection{No deterministic FPTAS}\label{subsec:no-fptas}

An FPTAS requires time polynomial in both the input length and $1/\varepsilon$. For any unweighted hedgegraph with $m$ hedges,
\begin{equation}\label{eq:wpc-range}
0\le\WPC(G)\le\FWPC(G)\le m.
\end{equation}

\begin{corollary}[No deterministic FPTAS]\label{cor:no-fptas}
Unless $\mathrm P=\mathrm{NP}$, neither objective admits a deterministic FPTAS, even on connected full-support split systems with $s(A)=3$. The same impossibility holds on the edge-colored simple graphs of Corollary~\ref{cor:simple-hardness}. This includes algorithms returning a feasible partition and value-approximation algorithms returning a rational number $z$ satisfying
\[
w\le z\le(1+\varepsilon)w,
\]
where $w$ is the corresponding optimum.
\end{corollary}
\begin{proof}
Put $N=|V|$ and assume $m\ge1$.
Let $w$ be either optimum, and let $z$ be the returned
value or the exactly evaluated score of the returned
partition.

For integer WPC, $w$ is an integer.
For fractional WPC, an optimal partition $\cP$ with
$p$ blocks gives $w=D_G(\cP)/(p-1)$, where $D_G(\cP)$
is an integer and $1\le p-1\le N-1$.
Thus, for either objective,
\[
\lfloor w\rfloor+1-w\ge\frac1{N-1}.
\]
Run a hypothetical deterministic FPTAS with
\[
\varepsilon=\frac1{2m(N-1)}.
\]
By~\eqref{eq:wpc-range},
\[
0\le z-w\le\varepsilon w\le\frac1{2(N-1)}.
\]
Hence $\lfloor z\rfloor=\lfloor w\rfloor=\WPC(G)$. Since all vertices and hedges are explicitly listed, $1/\varepsilon$ is polynomial in the input length. Either FPTAS would therefore compute integer WPC exactly in polynomial time on either hard instance class, implying $\mathrm P=\mathrm{NP}$.
\end{proof}

\section{Conclusions and open questions}\label{sec:conclusion}

We prove that the integer-threshold decision problem for WPC is NP-complete, answering the question raised in~\cite{r5}. The matrix formulation separates the polynomial-time case $s(A)\le2$ from the NP-hard case $s(A)=3$. The hardness constructions also establish a complexity separation between WPC and hedge connectivity.

For every fixed component bound $b$, our deterministic and randomized EPTASs return one partition simultaneously approximating integer and fractional WPC. Unless $\mathrm P=\mathrm{NP}$, neither objective admits a deterministic FPTAS even on either of the two hard instance classes.

Two questions remain. First, can the deterministic EPTAS on full-support split systems be improved to $2^{O(1/\varepsilon)}\poly(m,N)$, matching the dependence on $\varepsilon$ in the randomized algorithm? Second, what deterministic approximation guarantees are possible without a bound on the number of components per hedge, when an explicit vertex partition is required?

\end{document}